\documentclass[a4paper,11pt,fleqn]{article}
\usepackage[round]{natbib}

\usepackage{amsmath,amsthm}
\usepackage{amssymb}
\usepackage{dsfont}
\usepackage{amsfonts}
\usepackage{enumerate}
\usepackage{mathrsfs}
\usepackage{multirow}
\usepackage{float}
\usepackage{graphicx}
\usepackage{color}
\usepackage{abbreviationsSSC-ext}

\begin{document}

\title{On the eigenvalues of the spatial sign covariance matrix in more than two dimensions }
\author{Alexander D\"urre\footnote{corresponding author: alexander.duerre@udo.edu} \,   and David E.\ Tyler and Daniel Vogel}
\date{}
\maketitle

\vspace{-3.0ex}
\begin{center}
\footnotesize{
\noindent
Fakult\"at Statistik, Technische Universit\"at Dortmund, \\
44221 Dortmund, Germany 

\smallskip
\noindent
Department of Statistics \& Biostatistics, Rutgers University,  \\
Piscataway, NJ 08854, USA 

\smallskip
\noindent
Institute for Complex Systems \& Mathematical Biology, University of Aberdeen, \\
Aberdeen AB24 3UE, United Kingdom 
}
\end{center}

\begin{abstract} 
We gather several results on the eigenvalues of the spatial sign covariance matrix of an elliptical distribution. It is shown that the eigenvalues are a one-to-one function of the eigenvalues of the shape matrix and that they are closer together than the latter. We further provide a one-dimensional integral representation of the eigenvalues, which facilitates their numerical computation. 
\end{abstract}

{\footnotesize

}


\section{Introduction}

Let $X, X^{(1)}, \ldots, X^{(n)}$ be i.i.d.\ $p$-variate random vectors with distribution $F$ and let $\X_n = (X^{(1)}, \ldots, X^{(n)})^\top$ be the $n \times p$ data matrix. 
The sample spatial sign covariance matrix (SSCM) is defined as 
\[ 
	S_n	=	S_n(\X_n)=\frac{1}{n}\sum_{i=1}^n s(X^{(i)}-\mu_n)s(X^{(i)}-\mu_n)^\top,
\] 
where $\mu_n = \mu_n(\X_n) = \argmin_{\mbox{\scriptsize $\mu \in \R^p$}} \sum_{i=1}^n | X^{(i)} - \mu |$ is the spatial median of $\X_n$, and $s(x)$ denotes the spatial sign of $x \in \R^p$, i.e.,  $s(x) = x/|x|$ for $x \neq 0$ and $s(0) = 0$. Throughout, $|\cdot|$ denotes the Euclidean norm in $\R^p$.
Under very mild conditions on the data distribution $F$ \citep{duerre:vogel:tyler:2014}, the sample SSCM is strongly consistent for its population value
\[
	S(F) = E \left\{ s(X-\mu) s(X-\mu)^\top \right\}, 
\]
where $\mu = \mu(F) = \argmin_{\mbox{\scriptsize $\mu \in \R^p$}} E \left( |X - \mu| -|X|\right)$ is the population spatial median. 
In this note we study the population SSCM $S(F)$ under $p$-dimensional elliptical distributions. 
The random vector $X$ is said to be elliptical if it can be expressed as 
\be \label{eq:elliptical}
	X = A R U + \mu
\ee
for some $\mu \in \R$ and some arbitrary $p\times p$ matrix $A$, where $R$ is a univariate non-negative random variable, $U$ is uniformly distributed on the unit sphere in $\R^p$, and $R$ and $U$ are independent.
We impose the minimal regularity assumption that $A \neq 0$. For notational convenience, we presume $P(R = 0) = 0$.  The results, though, readily extend to the general case $P(R =0)  \geq 0$.
We call $V = A A^\top$ the shape matrix of $F$. For given $F$, the shape matrix $V$ is unique up to scale, and we understand the \emph{shape} of an elliptical distribution as an equivalence class of proportional, non-negative definite matrices. We use $F \in \Ee_p(\mu, V)$ to denote that $F$ is a $p$-dimensional distribution with shape matrix $V$ and central location $\mu$. The parameter $\mu$ coincides with the spatial median. If $A$ has full rank and $R$ has an absolutely continuous distribution, then $F$ possesses a Lebesgue-density $f$ in $\R^p$, which is of the form $f(x) = \det(V)^{-\frac{1}{2}}g((x-\mu)^\top V^{-1}(x-\mu))$ for some function $g:[0,\infty) \to [0,\infty)$. However, we do not make this assumption. 
We call $V_0 = V/\trace(V)$ the \emph{trace-normalized shape matrix}. It is popular to fix the shape by setting $\det(V) = 1$ \citep{paindaveine2008}, but we prefer to normalize the shape by the trace when comparing it to the SSCM, since by definition the latter has always trace one and so the corresponding sets of eigenvalues are of the same magnitude. All results given below hold true regardless of the specific choice of the normalization.

The SSCM $S(F)$ is given in terms of $V_0$ as follows. Let $V_0 = O \Lambda O^\top$ be an eigenvalue decomposition with $\Lambda = \diag(\lambda_1,\ldots,\lambda_p)$ with $\lambda_1 \ge \lambda_2 \ge \ldots \ge \lambda_p \ge 0$. Keep in mind that $\sum_j \lambda_j = \trace(V_0) = 1$. Then $S = O \Delta O^\top$ with $\Delta = \diag(\delta_1, \ldots, \delta_p)$, 
where
\be \label{eq:delta}
	\delta_i = E \left\{ \lambda_i Y_i^2 \left( \textstyle \sum\nolimits_{j=1}^p \lambda_j Y_j^2\right)^{-1} \right\}
\ee
and $Y = (Y_1,\ldots,Y_p) = R U$ has a spherical distribution. 
%
%
%
The asymptotic covariance matrix $W_S$ of the estimator $S_n$ at elliptical distributions can be written as 
\[
W_S = (O \times O) \left\{ \Gamma - \vec{\Delta} (\vec{\Delta})^\top \right\} (O \times O)^\top,
\] with
\[
	\Gamma = E \left\{ \vec \left( \frac{\Lambda^{1/2} Y Y^\top \Lambda^{1/2}}{Y^\top\Lambda Y}\right)  \vec \left( \frac{\Lambda^{1/2} Y Y^\top \Lambda^{1/2}}{Y^\top\Lambda Y}\right)^\top \right\}.
\]
Due to the spherical symmetry of $Y$, $p(p^3-3p +2)$ of the $p^4$ matrix entries of $\Gamma$ are zero. The remaining $p(3p-2)$ entries consist of at most $(p+1)p/2$ distinct values, with the upper bound being achieved if the eigenvalues $\lambda_1,\ldots, \lambda_p$ of $V_0$ are mutually distinct. 
Letting
\be \label{eq:eta}
	\eta_{ij} =  E \left\{ \lambda_i Y_i^2 \lambda_j Y_j^2 \left(\textstyle\sum_{j=1}^p \lambda_j Y_j^2\right)^{-2} \right\},  
	\qquad 1 \le i,j \le p,
\ee
we have for $1 \le i < j \le p$, each $\eta_{ij}$ appears six times in $\Gamma$, that is at the positions $\{(i-1)p+j,\,(i-1)p+j)\}$,
$\{(i-1)p+i,\,(j-1)p+j)\}$, $\{(i-1)p+j,\,(j-1)p+i)\}$, and the same with $i$ and $j$ interchanged. For $1 \le i \le p$, each $\eta_{ii}$ appears once at position $\{(i-1)p+i,\,(i-1)p+i)\}$.

It is important to note that $s(X)$ is invariant under the distribution of $R$. We thus have the liberty of choosing any specific spherical distribution for $Y$(which only has to satisfy $P(Y=0) = 0$) when analyzing the integrals (\ref{eq:delta}) and (\ref{eq:eta}). This observation is central for the results given below. 
%
%
%
%
%
%
%
%
Let 
\[
	\Phi_p = \left\{ x = (x_1, \ldots, x_p) \in [0,1]^p\, \middle| 1 \ge x_1 \ge x_2 \ge \ldots \ge x_p \ge 0, \ \sum\nolimits_{j=1}^p x_j = 1 \right\}
\]
and $\phi:\Phi_p \to \Phi_p$ denote the function that maps $(\lambda_1, \ldots, \lambda_p)$ to $(\delta_1, \ldots, \delta_p)$. The function $\phi$ is given by (\ref{eq:delta}). Below we give three results that relate the sets of eigenvalues $(\lambda_1,\ldots,\lambda_p)$ and $(\delta_1,\ldots,\delta_p)$. The proofs are deferred to the Appendix.

%
%
%
%
%

\section{On the eigenvalues of the spatial sign covariance matrix}

\begin{proposition} \label{lem:injective}
The function $\phi$ is injective, i.e., for any two vectors $\lambda, \tilde\lambda \in \Phi_p$, we have that $\phi(\lambda) = \phi(\tilde\lambda)$ implies $\lambda = \tilde\lambda$.
\end{proposition}
The content of Proposition \ref{lem:injective} is also mentioned in \cite{Magyar2014} and \citet{Vogel2015}. 
From (\ref{eq:delta}) it can be deduced that the $\delta_i$ obey the same ordering as the $\lambda_i$, i.e., $\delta_1 \ge \ldots \ge \delta_p \ge 0$, and furthermore that $\delta_i = \delta_j$ if and only if $\lambda_i = \lambda_j$, $1 \le i < j \le p$. The phrasing in \citet{Vogel2015} misleadingly suggests that Proposition \ref{lem:injective} can be readily concluded from this observation. \citet{Vogel2015} further note that the ``eigenvalues of $S(F)$ tend to be closer together than the eigenvalues of $V_0$''. 
Another way of phrasing this is to say, the ellipsoid associated with the SSCM is less elliptic or eccentric than that corresponding to the shape matrix. The next assertion makes this statement rigorous. 
\begin{proposition} \label{lem:closer.together} 
Let $1 \le i < j \le p$. If $\lambda_j > 0$, then $\delta_i/\delta_j \le \lambda_i/\lambda_j$. If additionally $\lambda_i > \lambda_j$, then the inequality is strict.
\end{proposition}
In the last proposition, we give one-dimensional integral representations for the eigenvalues $\delta_i$ of the SSCM and the expectations $\eta_{ij}$, cf.~(\ref{eq:eta}).
\begin{proposition} \label{lem:integral}
The integrals (\ref{eq:delta}) and (\ref{eq:eta}) can be expressed as
\[
	\delta_i = \frac{\lambda_i}{2}
	\int_0^\infty
	\frac{1}{(1 + \lambda_i x) \prod_{j=1}^p (1 + \lambda_j x)^{1/2}} d x,
	\qquad 1 \le i \le p, 
\]\[
	\eta_{ij}
	= \frac{\lambda_i \lambda_j}{4}
	\int_0^\infty
	\frac{x}{(1 + \lambda_i x) (1 + \lambda_j x)\prod_{j=1}^p (1 + \lambda_j x)^{1/2}} d x,
	\qquad 1 \le i, j \le p, \ i \neq j, 
\]\[
	\eta_{ii}
	= \frac{3 \lambda_i^2}{4}
	\int_0^\infty
	\frac{x}{(1 + \lambda_i x)^2 \prod_{j=1}^p (1 + \lambda_j x)^{1/2}} d x,
	\qquad 1 \le i \le p.
\]
\end{proposition}

%
%
%
%
%
%

\section{Applications}

Since the SSCM maintains the eigenvectors and the ordering of the eigenvalues of the shape matrix, it has been employed previously for analyses that rely on the orientation of the data only, most notably principal component analysis \citep[e.g.][]{marden1999, Croux2002, locantore1999, gervini2008}. 
Proposition \ref{lem:integral} considerably extends the applicability of the SSCM to analyses that require the full shape information. This includes all inference concerning multivariate dependencies under normality, for instance estimating correlations and partial correlations.
It has been noted before that, in dimension $p = 2$, the function $\phi$ and its inverse admit a simple analytic form, which allows to construct a 
robust pairwise correlation estimate based on the SSCM \citep{duerre:vogel:fried:2015, duerre:vogel:2016}.

The $p$-dimensional integrals (\ref{eq:delta}) and (\ref{eq:eta}) are unfeasible for numerical approximations. Solving the one-dimensional integrals given in Proposition \ref{lem:integral} numerically, the population SSCM as well as the asymptotic variance matrix of the sample SSCM can be computed in any dimension. Using the function integrate() in R \citep{R}, we found it to work without problems for $p = 10,\!000$.  Moreover, by means of a Newton-type algorithm, the function $\phi$ can be inverted, and one can hence construct a consistent shape matrix estimate solely based on the SSCM. Both, the function $\phi$ as well its inverse are implemented in the R-package sscor \citep{sscor}. Proposition \ref{lem:integral} provides a ``workable description of the eigenvalues of the $p$-dimensional spatial sign covariance matrix in terms of the eigenvalues of the covariance matrix'', which was identified as an open problem in \citet{duerre:vogel:2016}. 

Groups of eigenvalues of the shape matrix being well separated is the working assumption for principal component analysis.
Since the sample SSCM is inefficient in such a situation, the use of the SSCM for robust principal component analysis has also been questioned \citep[e.g.][]{Bali2011, Magyar2014}.
  In the case of two distinct eigenvalues, \citet{Magyar2014} investigate the asymptotic efficiency of the SSCM eigenspace projections by employing a representation of the eigenvalues $\delta_i$ and the $\eta_{ij}$-terms by means of the Gauss hyperbolic function. 
Proposition \ref{lem:integral} allows to quantify the asymptotic efficiency of the SSCM and any analysis build upon it in the general setting. 

It also important to note that the SSCM works for $p > n$. For many years, much attention has been paid to affine equivariant scatter estimation within the area of robust multivariate statistics. This paradigm, however, is unfeasible if $p > n$ \citep{Tyler2010}. With the increasing emergence of high-dimensional problems in recent years, there is also a renewed interest in non-affine equivariant estimators. 
The SSCM is an appealing candidate, not despite but because of its simplicity. Not only does it provide a robust scatter estimator in the $n < p$ setting, but also one which is computationally very feasible for large $p$.

%
%
%
%
%
%

An estimator related to the SSCM is the \emph{spatial Kendall's $\tau$ matrix}
\[
	K_n(\X_n) = 
 \frac{1}{\binom{n}{2} }\sum_{1\le i < j  \le n} s(X^{(i)}-X^{(j)}) s(X^{(i)}-X^{(j)})^\top,
\] 
which is the SSCM applied to the pairwise differences of the observations instead of the centered observations.
The spatial Kendall's $\tau$ matrix (or symmetrized SSCM) is computationally less appealing, has a lower breakdown point than the SSCM, it is not distribution-free under ellipticity, and its asymptotic variance has a more complex structure. However, it possesses a substantially higher efficiency at normality, which has lead several authors to consider this estimator in various contexts \citep{choi:marden:1998,taskinen:koch:oja:2012,fan:liu:wang:2015}.
By $U$-statistics theory, $K_n$ is asymptotically normal and consistent for
\[
	K(F) = E \{ s(X-Y)s(X-Y)^\top \}, \quad \qquad  X,Y \sim F, \mbox{ independent.}
\]
Since, for elliptical distributions $F$, the random vector $X-Y$ is elliptical with center zero and the same shape matrix as $X$, we have $S(F) = K(F)$, and any results concerning the eigenvalues of $S(F)$ also apply to $K(F)$.

Finally, our results concerning the spatial sign covariance matrix also extend to the \emph{generalized elliptical family}, as considered, e.g., by \citet[][Ch.~3]{Frahm2004}. The generalized elliptical model is generated by (\ref{eq:elliptical}) with
$A$, $R$, $U$ and $\mu$ as in (\ref{eq:elliptical}), except that $R$ and $U$ need not be independent. This class has 	been called \emph{distributions with elliptical direction} by \citet{randles:1989}.
Then $S(F)$ as well as the asymptotic covariance matrix of $S_n$ can be expressed  in terms of the eigenvalues of the shape matrix $V = A A^\top$ in the same way as under ellipticity. However, the proportionality between the shape matrix and the covariance matrix does not extend from the elliptical model to the generalized elliptical model.


\section*{Acknowledgments}

Alexander D\"urre was supported in part by the Collaborative Research Grant 823 of the German Research Foundation.   David E.\ Tyler was supported in part by the National Science Foundation grant DMS-1407751. A visit of Daniel Vogel to David E.\ Tyler was supported by a travel grant from the Scottish Universities Physics Alliance. The authors are grateful to the editors and referees for their constructive comments.

%
%
%
%
%
%
%
%


\appendix
\section{Proofs}

\begin{proof}[Proof of Proposition~\ref{lem:closer.together}]
If $\lambda_i = \lambda_j$, then $\delta_i = \delta_j$ by virtue of (\ref{eq:delta}). Assume $0 < \lambda_j < \lambda_i$. Using again (\ref{eq:delta}), the claim $\delta_i/\delta_j < \lambda_i/\lambda_j$ is equivalent to 
$E\{ Y_i^2/\sum_{k=1}^p \lambda_k Y_k^2 \}/ E\{ Y_j^2/\sum_{k=1}^p \lambda_k Y_k^2 \} < 1$, or equivalently
\[
	E \left\{ \frac{Y_j^2 - Y_i^2}{ Y_i^2 + (\lambda_j/\lambda_i) Y_j^2 + \sum_{k\neq i,j} (\lambda_k/\lambda_i) Y^2_k} \right\} > 0.
\]
Let $W = \sum_{k\neq i,j} (\lambda_k/\lambda_i) Y^2_k$, furthermore $X_1 = Y^2_i$, $X_2 = Y^2_j$ and $r = \lambda_j/\lambda_i$, hence $0 < r < 1$.
If we let $(Y_1,\ldots,Y_p)$ be standard normal, then $X_1, X_2 \sim \chi^2_1$ and $X_1, X_2, W$ are independent. The latter implies that the conditional distribution of $(X_1,X_2)$ given $W$ is the same as its unconditional distribution. Hence it suffices to prove
\be \label{eq:conditional}
	E \left\{ \frac{X_2-X_1}{r X_2 + X_1 + w} \right\} > 0 
\ee
for every $w > 0$. We apply the orthogonal transformation 
\[
	\begin{pmatrix}
		Z_1 \\
		Z_2
	\end{pmatrix}	
	=	\frac{1}{\sqrt{2}}
		\begin{pmatrix}
			1 & 1 \\
			-1 & 1 \\
	\end{pmatrix}	
		\begin{pmatrix}
		X_1 \\
		X_2
	\end{pmatrix}	
, \quad \mbox{  i.e., } \quad
		\begin{pmatrix}
		X_1\\
		X_2
	\end{pmatrix}	
	=	\frac{1}{\sqrt{2}}
		\begin{pmatrix}
			1 & -1 \\
			1 & 1 \\
	\end{pmatrix}	
		\begin{pmatrix}
		Z_1 \\
		Z_2
	\end{pmatrix}	
\]
and obtain that the expectation on the left-hand side of (\ref{eq:conditional}) is equal to 
\[
	 2 E \left\{ \frac{Z_2}{ (r+1) Z_1  + (r-1) Z_2 + c} \right\}  
	= 2 E \left[ E \left\{ \frac{Z_2}{ (r+1) Z_1  + (r-1) Z_2 + c}\, \middle|\, Z_1 \right\} \right]
\]
with $c = \sqrt{2}w$.
The inner integral is 
\[
	I(z_1) = 
	\int_{-z_1}^{z_1} \frac{z_2 f(z_2|z_1)}{ (r+1) z_1  + (r-1) z_2 + c}  d z_2,
\]
where $f(z_2|z_1)$ denotes the conditional density of $Z_2$ given $Z_1 = z_1$. Since $(X_1,X_2) = (X_2, X_1)$ in distribution, we have that, for every $z_1 \ge 0$, the conditional distributions given $Z_1 = z_1$ of $Z_2$ and $-Z_2$ are the same, and hence $f(z_2|z_1) = f(-z_2|z_1)$. Therefore,
\[
	I(z_1) = 
	\int_{-z_1}^0 \frac{z_2 f(z_2|z_1)}{ (r+1) z_1  + (r-1) z_2 + c}  d z_2 + 
	\int_0^{z_1} \frac{z_2 f(z_2|z_1)}{ (r+1) z_1  + (r-1) z_2 + c} d z_2
\]
\[
	\qquad  = 
	\int_0^{z_1} \left( \frac{z_2 f(z_2|z_1)}{ (r+1) z_1  + (r-1) z_2 + c} + \frac{-z_2 f(z_2|z_1)}{ (r+1) z_1  - (r-1) z_2 + c}\right)  d z_2
\]
\[
	\qquad  = 
	\int_0^{z_1}  \frac{2 (1-r) z^2_2 f(z_2|z_1)}{ \{ (r+1) z_1  + c\}^2  - (1-r)^2 z_2^2 }  d z_2.
\]
The integrand is strictly positive over the integration domain. Hence (\ref{eq:conditional}) holds, and the proof is complete. 
\end{proof}
\begin{proof}[Proof of Proposition~\ref{lem:injective}]
Let $\lambda$, $\tilde\lambda \in \Phi_p$ and $\lambda \neq \tilde\lambda$. We show that $\phi(\lambda) = \phi(\tilde\lambda)$ can not hold. The latter would be equivalent to 
\be \label{eq:injective1}
	 E \left\{ \lambda_i Y_i^2 \left( \textstyle \sum\nolimits_{j=1}^p \lambda_j Y_j^2\right)^{-1} \right\}
	= E \left\{ \tilde\lambda_i Y_i^2 \left( \textstyle \sum\nolimits_{j=1}^p \tilde\lambda_j Y_j^2\right)^{-1} \right\}
	\quad \mbox{ for all } 1 \le i \le p.
\ee
If there is a $k \in \{ 1,\ldots, p \}$ such that $\lambda_k = 0$ and $\tilde\lambda_k \neq 0$ or vice versa, then (\ref{eq:injective1}) is clearly violated. Thus we have $\lambda_k = 0 \Leftrightarrow \tilde\lambda_k = 0$ for all $1 \le k \le p$. Next, let $\rho_k = \lambda_k/\tilde\lambda_k$ for $\tilde\lambda_k \neq 0$ and $\rho_k = 0$ otherwise, and define $Z_k = \tilde\lambda_k Y_k^2$, $1 \le k \le p$. Then (\ref{eq:injective1}) can be expressed as
\be \label{eq:injective2}
	E \left\{ Z_i \left( \textstyle \sum\nolimits_{j=1}^p Z_j\right)^{-1} \right\} 
	=
	E \left\{ \rho_i Z_i \left( \textstyle \sum\nolimits_{j=1}^p \rho_j Z_j\right)^{-1} \right\} 
	\quad \mbox{ for all } 1 \le i \le p.
\ee
Let $\rho_{(1)}, \ldots, \rho_{(p)}$ be the numbers $\rho_1, \ldots, \rho_p$ arranged in ascending order. Note that $\rho_{(1)} = \rho_{(p)}$ is only possible if $\lambda \propto \tilde\lambda$, which contradicts our assumptions. Hence $\rho_{(1)} < \rho_{(p)}$. Let $k_0$ be such that $\rho_{k_0} = \rho_{(p)}$.
Since $Y = (Y_1, \ldots, Y_p)$ is spherical with $P(Y=0) = 0$, we have $\sum_{j=1}^p \rho_j Z_j < \rho_{k_0} \sum_{j=1}^p  Z_j$ almost surely, and consequently 
\[
	E \left\{ \rho_{k_0} Z_{k_0} \left( \textstyle \sum\nolimits_{j=1}^p \rho_j Z_j\right)^{-1} \right\}
	> 
	E \left\{ \rho_{k_0} Z_{k_0} \left( \rho_{k_0} \textstyle \sum\nolimits_{j=1}^p Z_j\right)^{-1} \right\}
	= 
	E \left\{ Z_{k_0} \left( \textstyle \sum\nolimits_{j=1}^p Z_j\right)^{-1} \right\}, 
\]
and so (\ref{eq:injective2}) and hence (\ref{eq:injective1}) can not hold. 
\end{proof}
\begin{proof}[Proof of Proposition~\ref{lem:integral}]
We exercise the liberty to choose an appropriate distribution for $Y$ and take the uniform distribution on the unit ball (not the unit sphere) with density
	$f(y)= p/2 \Gamma(p/2)\pi^{-p/2} \mathds{1}_{[0,1]}(y^\top y)$,
which yields 
\[
	\delta_i 
	\ =\ \frac{p\Gamma(p/2)}{2\pi^{p/2}}\int_{y^\top y \le 1}\frac{\lambda_iy_i^2}{\lambda_1y_1^2+\ldots+\lambda_py_p^2}dy
	\ =\ \frac{2^p p\Gamma(p/2)}{2\pi^{p/2}} \int_{S_2} \frac{\lambda_iy_i^2}{\lambda_1y_1^2+\ldots+\lambda_py_p^2}dy
\]
with $y=(y_1,\ldots, y_p)$ and $S_{2,p} = \{ y \in \R^p \,|\, y^\top y \le 1, y \ge 0 \}$.  
Substituting $y_ k=\sqrt{z_k}$, $1 \le k \le p$,  we get
\be \label{eq:substitution}
	\delta_i
	= \frac{ p\Gamma(p/2)}{2\pi^{p/2}}
	  \int_{S_{1,p} } \frac{\lambda_i z_i}{\lambda_1 z_1+\ldots+\lambda_p z_p} \prod_{j=1}^p\frac{1}{\sqrt{z_j}} dz
\ee
with $S_{1,p} = \{ z \in \R^p \,|\, z_1 + \ldots + z_p \le 1, z \ge 0 \}$. Now we apply formula 4.646 in \citet{Gradshteyn2007}:
\be \label{eq:integralformel}
	\int_{S_{1,n}} \frac{\prod_{k=1}^n x_k^{p_k-1}}{(\sum_{k=1}^n q_k x_k)^r}dx 
	= 
	\frac{\Gamma(p_1)\cdot \ldots \cdot \Gamma(p_n)}{\Gamma(\sum_{k=1}^n p_k - r+1)\Gamma(r)}
		\int_0^\infty \frac{x^{r-1}}{\prod_{k=1}^n (1+q_k x)^{p_k}}dx
\ee
for $p_1,\ldots,p_n,q_1,\ldots q_n > 0$, $p_1+\ldots+p_n>r>0$. Setting $n = p$, $r =1$, $q_k = \lambda_k$ for $1 \le k \le p$, $p_i = 3/2$, and $p_k = 1/2$ for $k \neq i$, we obtain from (\ref{eq:substitution}) the expression for $\delta_i$ given in Proposition \ref{lem:integral}.
As for $\eta_{ii}$, we proceed similarly. Choosing again the uniform density on the unit ball and substituting $y_ k=\sqrt{z_k}$, $1 \le k \le p$, yields
\[
	\eta_{ii} 
	=\frac{p\Gamma(p/2)}{2\pi^{p/2}} \int_{S_{1,p}}
	\frac{\lambda_i^2z_i^2}{(\lambda_1 z_1 +\ldots+\lambda_p z_p)^2}\prod_{j=1}^p\frac{1}{\sqrt{z_j}}d z.
\]
Applying (\ref{eq:integralformel}) with $n = p$, $r =2$, $q_k = \lambda_k$ for $1 \le k \le p$, $p_i = 5/2$, and $p_k = 1/2$ for $k \neq i$, we obtain the expression for $\eta_{ii}$ as given in Proposition \ref{lem:integral}. As for $\eta_{ij}$ with $i \neq j$, we obtain a similar expression, to which we apply (\ref{eq:integralformel}) with the same parameters except $p_i = p_j = 3/2$, and $p_k = 1/2$ for $k \neq i,j$. This completes the proof.
\end{proof}

\bibliographystyle{abbrvnat}
{\small

}

\end{document}